\pdfoutput=1
\documentclass{ws-ijbc}
\makeatletter\def\wsdraftnote{}\makeatother
\usepackage{xcolor}
\usepackage[verbose]{hyperref}
\hypersetup{hidelinks}
\newtheorem{assumption}{Assumption}

\newcommand{\R}{\mathbb{R}}
\newcommand{\C}{\mathbb{C}}
\newcommand{\N}{\mathbb{N}}
\newcommand{\Sone}{\mathbb{S}^1}
\newcommand{\Spec}{\operatorname{Spec}}
\newcommand{\Int}{\operatorname{Int}}
\newcommand{\Rel}{\operatorname{Re}}
\newcommand{\Wc}{\mathcal{W}}

\begin{document}

\markboth{N. B. Tufillaro}{Memory as Motion on a Slow Invariant Manifold}

\title{Memory in Behavioral Models as Motion on a Slow Invariant Manifold}

\author{Nicholas B. Tufillaro}
\address{Aqualytics, 311 SW 2nd St., 1028\\
Corvallis, Oregon 97333, USA\\
nick@aqualytics.eco}

\maketitle

\begin{history}
\end{history}

\begin{abstract}
A large-signal operating point of a nonlinear two-port is a periodic orbit of a
periodically forced circuit. When the device has memory, the Floquet
exponents of that orbit separate into fast (electrical) and slow (thermal and trapping)
modes, and long-term memory is motion on the invariant manifold attached to the slow
modes. The parameterization theorem of Cabr\'e, Fontich and de la Llave, applied to the
stroboscopic map at the orbit, gives that manifold and its uniqueness among analytic
manifolds; it is the spectral submanifold of Haller and Ponsioen, obtained without a
small-forcing parameter. The theorem is theirs; added here are its reading for
behavioral models and a worked device. An exact reduced model has one memory state per
slow exponent, the memoryless model is the fixed-point family of the reduced dynamics,
and the step responses used to identify dynamic X-parameters are flow lines that start on
that family. The hypotheses are examined and the manifold constructed for a GaN
transistor compact model with thermal and trap memory, in simulation; no measurement is
reported. The trap's time constant at the operating point, about $15\,\mu$s, is set by
the linearization, and its $6$\,ms emission time does not appear in the spectrum. The
trap's nonlinearity confines a polynomial parameterization of the manifold to a few
millivolts of drive, so trap memory must be tabulated over the operating range. The reduced trap dynamics, its nonlinearity included, agrees with
transient simulations of the full circuit to within $1\%$.
\end{abstract}


\section{Introduction}

A behavioral model of a microwave device is an input--output model
identified from measurements at its terminals and used in circuit design in place of a
physical model. For a power transistor the inputs and outputs are the incident and
scattered waves at the two ports, and the hard case is a device with memory: self-heating
and charge trapping make the response to a modulated drive depend on its history over
microseconds to milliseconds, a million times longer than the period of the carrier. This
paper identifies the dynamical object behind such models, a spectral submanifold of the
forced circuit, and works it out for a transistor model.

A frequency-domain behavioral model fixes a harmonic grid and describes
a device by the map from the incident-wave phasors to the scattered-wave phasors. For a device
operating in a single-tone continuous-wave steady state, the form of that map is fixed by time
invariance. A companion paper~\cite{Tufillaro2026a} collects the argument that the map is equivariant
under a weighted action of the circle group. Practical forms of the map include the Cardiff
model~\cite{Woodington2010} and the X-parameter framework~\cite{Root2005,Root2013},
both of which take the drive phase as the phase reference. The model form follows from time invariance alone, so it
holds whether or not the device has memory; memory shows itself only when the drive is
modulated. The memoryless model rests on a geometric object, the map equivariant under
the circle group. This paper asks what geometric object underlies a dynamic behavioral
model of a device with memory under a modulated drive.

Two engineering answers exist. \citet{Verspecht2009} extend X-parameters with a kernel integrated over the
history of the envelope (the slowly varying amplitude of the drive), identified from the
transient of the output after a step change of that amplitude. Alternatively, time-domain approaches~\cite{Wood2004,Wood2005} reconstruct a state
from delayed samples of the envelope and fit a map on the reconstructed state, with the
embedding dimension found empirically. The dynamic gain model of \citet{Verspecht2023} lies between the two: memory enters through fixed difference
operators on delayed samples of the envelope, linearly, with coefficients that depend on
the current envelope level, and its authors name the choice of which operators to include,
and how many, as the open problem. Models that carry explicit slow states exist as
well: the hidden-variable formulation of the X-parameter book~\cite[Sec.~6.6.2]{Root2013},
in which an unspecified number of ``hidden'' dynamical variables $h_i(t)$ (temperatures,
bias voltages, trap states) enter the memoryless map, each with a quasi-static value
$s_i(|A|)$ about which the map is linearized, and the dynamic X-parameter kernel
of \citet{Verspecht2009}, which identifies that linearization. Others are state-space models with added slow states identified from multisine
measurements~\cite{Schreurs2003}, the two-time-scale dynamic Volterra models of Ngoya and
co-workers~\cite{Ngoya2009,Soury2012}, and compact models with explicit trap and thermal
states extracted from large-signal measurements~\cite{Xu2014}. In each of these the
number of slow states is chosen by the modeler. These answers leave open how many states a memory model
needs, what the identified kernel or map represents, and when a model built at one
operating point can be trusted at another.

A dynamical-systems view gives a third answer,
worked out by \citet{Haller2016}, who defined the spectral
submanifold (SSM) of a periodically forced system: the unique smoothest invariant manifold tangent to a chosen
set of modes of the linearized dynamics, with reduced dynamics that is the exact
nonlinear continuation of those modes. Their existence theorem is stated for a small
periodic forcing added to a stable fixed point; the general result behind it, due to
\citet{Haro2006} and, at fixed points, to \citet{Cabre2003}, is not perturbative: it applies to the map at the
orbit itself and does not require the forcing to be small.

This paper puts the engineering and dynamical-systems answers together. A large-signal operating point (LSOP) of a
device with memory is a periodic orbit whose Floquet exponents fall into fast
(electrical) and slow (thermal and trapping) modes separated by many orders of magnitude. Applying
the fixed-point theorem of \citet{Cabre2003} to the stroboscopic map of the circuit
at that orbit gives, without any smallness parameter, the invariant manifold attached to
the slow modes (Theorem~\ref{thm:main}): a bundle over the circle of drive phase whose
fiber dimension is the number of slow Floquet exponents.

The reduced dynamics on the slow manifold is the
object that an envelope-domain model\footnote{An envelope-domain model is a mixed
time--frequency description used in circuit simulation when the dynamics has two widely
separated time scales, a carrier and a slowly varying envelope: the carrier is handled in
the frequency domain by harmonic balance and the envelope is integrated in
time~\cite[Secs.~2.4 and 4.10]{Pedro2018}.} captures,
and the three questions left open above are answered in terms of it. The number of
memory states of an exact reduced model is the number of slow exponents
(Corollary~\ref{cor:count}). The step response from which the dynamic X-parameter kernel
of \citet{Verspecht2009} is identified is a flow line of the reduced dynamics from one
of its fixed points to another, a sum of
exponentials only in the small-step limit (Section~\ref{sec:kernel}). The map of a
time-domain model is the same reduced dynamics in delay coordinates. Stark's
delay-embedding theorem for forced systems~\cite{Stark1999} gives, at constant
drive, a sufficient number of delays: twice the number of slow
exponents plus one (Section~\ref{sec:bundle}). A model built at one operating point extends to a drive that
moves between operating points under an adiabatic hypothesis that is stated in
Proposition~\ref{prop:modulated}. Two further facts connect the slow manifold with the memoryless
theory: the memoryless X-parameter surface is the family of fixed points of the reduced
dynamics, and the slow manifold inherits the equivariance of the companion paper~\cite{Tufillaro2026a}
(Section~\ref{sec:fixedpoints}).

Most of the mathematics used here is restated. Theorem~\ref{thm:main} is the theorem
of \citet{Cabre2003} for the stroboscopic map; the manifold is a spectral submanifold
in the sense of \citet{Haller2016}; the statement in delay coordinates is Stark's
theorem; and the statement for a modulated drive is the periodic-orbit counterpart of the
adiabatic spectral submanifolds of \citet{Haller2024}. What is
added is the identification of the objects of behavioral modeling with objects on the
manifold (the state count, the fixed-point family, the step responses behind the kernel,
the envelope-domain model) and the worked device of Section~\ref{sec:example}, with its
trap exponent, its resonances and the range of its polynomial representation. For the
theory of spectral submanifolds the device is a test case of some interest on four counts.
The manifold is anchored at a strongly forced orbit, with a gap of more than four decades
between its slow and fast exponents and a spectral quotient of $10^{7}$. Its polynomial
parameterization has a range set by an exponential nonlinearity, a few millivolts of
drive, so that the reduced dynamics has to be tabulated globally. A sub-manifold of its
slow modes passes through low-order resonances as a parameter is swept. And the reduced
dynamics along the dominant, strongly nonlinear direction is checked against transient
simulations of the full circuit (Section~\ref{sec:fullcheck}). Slow invariant manifolds of
electronic circuits have so far been studied mainly in autonomous models, Chua's circuit
and its memristor-based extensions among them, by curvature and singular-perturbation
methods \cite{Ginoux2008,Ginoux2013,Ginoux2024};
here the manifold is attached to a strongly forced periodic orbit. The example
is a simulation of a compact model; no measurement is reported.

The second half of the paper examines the hypotheses of Theorem~\ref{thm:main} on a compact model of a GaN
high-electron-mobility transistor (HEMT), the ASM-HEMT~\cite{Khandelwal2018} with a three-pole thermal network and a
drain-lag trap, run in an open-source circuit simulator with the slow states handled by
a two-timescale reduction (Section~\ref{sec:example}). In nonlinear electronic circuits
more generally, infrared records of the temperature of the integrated devices show that
chaotic regimes heat more slowly than periodic ones \cite{Buscarino2020}. Three results come out of the
examination. The non-resonance conditions of the theorem are generic for the full set of
slow modes. For a model that keeps only the slowest thermal mode the condition at order 15
holds by a margin of only $7\times10^{-5}$, and for the two slowest thermal modes together
the trap mode passes through resonances of low order, five to seven, as the drive is varied.
The trap contributes a Floquet exponent with a time constant of about $15\,\mu$s, which is
a property of the linearization; the trap's emission time, the $6$\,ms time constant with
which it releases its charge, does not appear. The polynomial parameterization the theorem produces at
the operating point is adequate for the thermal submanifolds but valid only within a few
millivolts of drive for the trap.

The practical conclusion is that the geometric object underlying a memory model is the
one Theorem~\ref{thm:main} describes. The
operating point is a closed curve in the state space of the circuit, traced once per
drive period. Through each point of that curve passes a small surface, of dimension
equal to the number of slow modes, spanned by the slow directions, here the device
temperatures and the trap voltage. As the drive phase advances, the orbit carries that
surface around with it, and the surface sweeps out a tube wrapped around the orbit.
Every nearby trajectory falls onto the tube at the electrical rates and then winds along
it, once around per drive period, approaching the orbit at the thermal and trapping
rates. The tube is the invariant
manifold of the theorem, attracting and unique among analytic manifolds; the orbit alone is the memoryless
operating point; and memory is the slow motion on the tube. For the trap that motion
must be tabulated or fitted over the operating range; a Taylor
expansion about the orbit does not describe it.

\section{Setting}
\label{sec:setting}

\subsection{The circuit as a periodically forced analytic system}

Let $z\in\R^N$ collect the state variables of the device and its embedding circuit ---
node charges, inductor fluxes, and the internal memory states (junction temperatures,
trap occupancies). A time-invariant device driven at one frequency $\omega=2\pi f_0$
through fixed terminations obeys
\begin{equation}
  \dot z = G\bigl(z,\ \Rel(A\,e^{j\omega t})\bigr) =: G(z,\theta;A),\qquad \theta=\omega t,
  \label{eq:ode}
\end{equation}
where $A=(A_1,A_2)\in\C^2$ are the incident waves at the two ports and $G$ is analytic
in $z$ and in the drive on the operating region and $2\pi$-periodic in $\theta$.\footnote{Compact models
are built to be smooth; the ASM-HEMT used in Section~\ref{sec:example} is analytic on
its operating region, and the thermal and trap networks added to it are analytic by
construction. Piecewise definitions, where present, are clamps inactive at the operating
points considered.} Because the device is time invariant, $\theta$ enters only through
the instantaneous drive, so
\begin{equation}
  G\bigl(z,\theta;\,e^{j\varphi}A\bigr) = G(z,\theta+\varphi;A)
  \qquad\text{for all }\varphi\in\Sone :
  \label{eq:equivariance}
\end{equation}
the phase-shift equivariance of the companion paper~\cite{Tufillaro2026a}, here on the extended
phase space $(z,\theta)\in\R^N\times\Sone$.

\subsection{The LSOP as a periodic orbit}

A single-tone LSOP is a $T$-periodic solution $Z(\theta;A)$ of \eqref{eq:ode},
$T=2\pi/\omega$. Let $\Phi_t(\cdot\,;A)$ be the flow of \eqref{eq:ode} from $\theta=0$
and $P_A:=\Phi_T(\cdot\,;A)$ the stroboscopic map, an analytic diffeomorphism on the
operating region. The LSOP is a fixed point $z^\ast=Z(0;A)$ of $P_A$, and
$M:=DP_A(z^\ast)$ is its monodromy matrix, with eigenvalues the Floquet multipliers
$\mu_i$ and exponents $\lambda_i=T^{-1}\log\mu_i$, defined modulo $j\omega$; only their
real parts enter below. The Floquet multipliers of a harmonic-balance steady state are
the objects of the stability analysis of nonlinear microwave circuits, computed there from
the conversion matrix of the periodic linearization~\cite{Rizzoli1985,Mons1999,Suarez2003};
the slow Jacobian of Section~\ref{sec:spectrum} is the leading-order, period-averaged
form of that machinery, restricted to the slow states.

\begin{assumption}[Asymptotic stability]\label{ass:stable}
$|\mu_i|<1$, i.e.\ $\Rel\lambda_i<0$, for $i=1,\dots,N$.
\end{assumption}

Assumption~\ref{ass:stable} states what an operating point is in practice: a periodic
steady state that a harmonic-balance simulation converges to and that a measurement
settles on once transients have died out. Since $1\notin\Spec M$, the assumption also
implies, by the implicit function theorem, that the LSOP persists and depends
analytically on $A$ near any drive at which it exists.

\subsection{Slow modes and the spectral quotient}

Order the exponents by decay rate, $0>\Rel\lambda_1\ge\dots\ge\Rel\lambda_N$. In a
device with memory they fall into two groups separated by many orders of magnitude: the
slow modes (thermal and trap, $10^{2}$--$10^{6}\,\mathrm{s}^{-1}$), with which memory is
associated, and the fast modes (electrical, $10^{10}\,\mathrm{s}^{-1}$ and beyond). Let
$E\subset\R^N$ be the spectral subspace of $M$ for the $q$ multipliers of the slow
modes, and $E^c$ the complementary spectral subspace.

\begin{definition}[Absolute spectral quotient; \citealp{Haller2016}]
\begin{equation}
  \Sigma(E):=\Int\!\left[\frac{\min_{i}\Rel\lambda_i}{\max_{\lambda_i\in\Spec(M|_E)}\Rel\lambda_i}\right],
\end{equation}
the integer part of the ratio of the fastest decay rate in the circuit to the slowest
decay rate in $E$.
\end{definition}

\begin{assumption}[Non-resonance of the slow modes]\label{ass:nonres}
For every exponent $\lambda_l$ outside the slow modes and every $m\in\N^q$ with
$2\le|m|\le\Sigma(E)$,
\begin{equation}
  \langle m,\Rel\lambda_E\rangle:=\sum_{i=1}^{q}m_i\Rel\lambda_i\;\ne\;\Rel\lambda_l .
  \label{eq:nonres}
\end{equation}
\end{assumption}

Assumption~\ref{ass:nonres} is the form given in \citet{Haller2016}; it is sufficient for the condition used
in the proof below, that no product of at most $\Sigma(E)$ multipliers from $E$ equals a
multiplier outside $E$, since differing real parts of exponents give differing moduli
of multipliers.

\section{The Slow Manifold of an LSOP}
\label{sec:theorem}

\begin{theorem}[\citealp{Cabre2003}, for the stroboscopic map of an LSOP]\label{thm:main}
Let $G$ in \eqref{eq:ode} be analytic in $z$ and let the LSOP $Z(\cdot\,;A)$ satisfy
Assumptions~\ref{ass:stable} and~\ref{ass:nonres}. Then:
\begin{enumerate}
\item[(a)] There is a $q$-dimensional analytic manifold $W_0\ni z^\ast$, tangent to $E$
  at $z^\ast$ and locally invariant under $P_A$.
\item[(b)] $W_0$ is unique among locally invariant manifolds tangent to $E$ at $z^\ast$
  of class $C^{\Sigma(E)+1}$; in particular it is the unique analytic one.
\item[(c)] With $W_\theta:=\Phi_{\theta/\omega}(W_0;A)$, the set
  $\Wc(A):=\{(z,\theta)\in\R^N\times\Sone: z\in W_\theta\}$ is a $(q{+}1)$-dimensional
  analytic manifold, locally invariant under the extended flow
  $(\dot z,\dot\theta)=(G,\omega)$, containing the LSOP, and fibering over the circle of
  drive phase with fiber $W_\theta$ tangent to the slow Floquet subspace at
  $Z(\theta;A)$.
\item[(d)] $P_A|_{W_0}$ is analytically conjugate to a polynomial map $R:E\to E$ with
  $R(0)=0$, $DR(0)=M|_E$, of degree at most $\Sigma(E)$; $R$ may be taken linear if
  there are no internal resonances $\langle m,\Rel\lambda_E\rangle=\Rel\lambda_i$,
  $\lambda_i\in\Spec(M|_E)$, $2\le|m|\le\Sigma(E)$, and of degree at most $M_0-1$ if
  no internal resonance has order $|m|\ge M_0$.
\end{enumerate}
\end{theorem}

\begin{proof}
The theorem is Theorem~1.1 of \citet{Cabre2003} applied to $F=P_A$ at $z^\ast$ with
$X=\R^N$, $X_1=E$, $X_2=E^c$, in the analytic class. Its hypotheses, numbered as there, are
satisfied as follows: (0) $P_A$ is a
local diffeomorphism, being the time-$T$ map of an analytic flow, so $M$ is invertible;
(1) $E$ is $M$-invariant, and since $E^c$ is also invariant the off-diagonal block of
its triangular form vanishes; (2) $\Spec(M|_E)$ lies in the open unit disk by
Assumption~\ref{ass:stable}; (3) $0\notin\Spec(M|_{E^c})$. The integer $L$ of that theorem must
satisfy $\Spec(M|_E)^{L+1}\Spec(M^{-1})\subset\{|z|<1\}$; the worst case pairs the
slowest multiplier of $E$ with the fastest of the whole spectrum, so the condition
reads $(L+1)\max_E\Rel\lambda<\min\Rel\lambda$, i.e.\
$L+1>\min\Rel\lambda/\max_E\Rel\lambda$, which $L=\Sigma(E)$ satisfies. (4) The
non-resonance $\Spec(M|_E)^{i}\cap\Spec(M|_{E^c})=\emptyset$ for $2\le i\le L$ follows
from Assumption~\ref{ass:nonres}. (5) $L+1\le r$, with $r$ the differentiability order there, holds
since $F$ is analytic, for which that order is unbounded. The
conclusions of that theorem give (a), (d) and (b): an analytic $K:U_1\subset E\to\R^N$ with
$P_A\circ K=K\circ R$, $K(0)=z^\ast$, $DK(0)$ the inclusion of $E$, $R$ polynomial of
degree $\le L$, linear under its condition (7) and of degree $\le M_0-1$ under its
condition (8), and uniqueness of the manifold $K(U_1)$ among $C^{L+1}$ locally invariant
manifolds tangent to $E$. For (c): $\Phi_t$ maps $W_\theta$ into $W_{\theta+\omega t}$ by
definition and $W_{2\pi}=P_A(W_0)\subset W_0$ by (a), so the family closes up over the
circle; each $W_\theta$ is the image of $W_0$ under an analytic diffeomorphism depending
analytically on $\theta$, $z^\ast\in W_0$, and $T_{Z(\theta)}W_\theta=D\Phi_{\theta/\omega}(z^\ast)E$
is the slow Floquet subspace transported along the orbit.
\end{proof}

We want to emphasize three points about Theorem~\ref{thm:main}. First, on the relation
to the SSM theorems: \citet[Thm.~4]{Haller2016} treat periodic
forcing as an $\varepsilon$-small perturbation of an autonomous system with a stable
fixed point, proved through the invariant-torus theorem of
\citet{Haro2006}. An LSOP is not a small perturbation of a bias point, but
for one-frequency forcing no perturbation argument is needed: the torus is the periodic
orbit, its stroboscopic map has a fixed point, and the fixed-point theorem applies
directly. Theorem~\ref{thm:main} is the SSM theorem in its one-frequency form, stated at the
orbit one simulates or measures; the non-resonance condition and spectral quotient are
exactly those of \citet{Haller2016}. The parameterization of invariant manifolds
attached to periodic orbits, with the Floquet normal form supplying the bundle structure
of (c), is developed by \citet{Castelli2015,Castelli2018},
and the computational side of the method in general is presented in \citet{Haro2016};
methods for computing global stable and unstable manifolds of vector fields are surveyed
in \citet{Krauskopf2005};
what is done here is to apply that construction to the operating point of a device.

Second, on attraction: for the slow modes --- $E$ spanned by the $q$ slowest modes ---
every rate transverse to $E$ is strictly faster than every rate within it. $W_0$ is then
normally attracting. A closed neighborhood of the fixed point in $W_0$ is overflowing
invariant under $P_A^{-1}$, with the fast directions expanding, and the applicable
statement is Fenichel's theorem on asymptotic stability with rate
conditions~\cite[Theorem~5]{Fenichel1974}: each trajectory near the LSOP approaches the
trajectory of one point of $\Wc(A)$, its asymptotic phase, at the fast rates, and then
the LSOP within $\Wc(A)$ at the slow rates. This is the property for which slow
spectral submanifolds are used in model reduction~\cite{Cenedese2022,Haller2023}. Attraction is why $\Wc(A)$ is the object a
measurement observes: once the electrical transient has decayed, the trajectory lies on
$\Wc(A)$.

Third, on the smoothness class: for a realistic device $\Sigma(E)$ is of order $10^{7}$
(Section~\ref{sec:example}), so uniqueness in (b) is asserted among $C^{10^7}$ manifolds.
For an analytic compact model this is of no consequence. For a manifold inferred from measured
data it is vacuous: the invariant manifolds tangent to $E$ of lower smoothness form a
family, described in \citet{Haller2023}, that data cannot tell apart from $W_0$.
A statement that a data-driven slow manifold is the right one must rest on
the attraction of the second point.

\section{Consequences for Behavioral Models}
\label{sec:consequences}

\subsection{The number of memory states}

\begin{corollary}\label{cor:count}
Under the hypotheses of Theorem~\ref{thm:main}, the reduced dynamics on $W_0$ has
$q=\dim E$ states, the number of slow modes. Within that reduced dynamics,
Theorem~\ref{thm:main} gives a further reduced model on a proper spectral subspace
$E'\subsetneq E$ when \eqref{eq:nonres} holds with the omitted slow modes in the role of
the outside exponents and with the quotient $\Sigma$ computed within the slow spectrum.
Its uniqueness is, as in Theorem~\ref{thm:main}(b), among manifolds of class $C^{\Sigma+1}$.
\end{corollary}

The first part of the corollary restates the dimension of the manifold. Its content is
that the manifold exists and that $q$ is fixed by the gap in the spectrum, so that the
answer to ``how many memory states does a dynamic X-parameter model need'' is the number
of slow Floquet exponents at the operating point. The second part, on reduced models that
keep only some of the slow modes, is the one with practical consequences. Between the
slow modes and the fast modes the ratio of rates is $3\times10^{4}$ or more in the example
of Section~\ref{sec:example}, and \eqref{eq:nonres} is a coincidence condition at orders
from $3\times10^{4}$ to $10^{7}$ and beyond, generically true. Among the slow modes the
ratios are small numbers, and close to integers when thermal and trap fits use round time
constants. Section~\ref{sec:example} exhibits a ratio of $15$ to within $7\times10^{-5}$,
and for a model that keeps only the slowest mode the condition at that order holds by that
margin only. The condition is sufficient: where it fails, an invariant manifold tangent to
$E'$ may still exist with lower smoothness~\cite{Haller2023}.

\subsection{The memoryless model is the fixed-point family}
\label{sec:fixedpoints}

By Assumption~\ref{ass:stable} and the analyticity of $G$ in the drive,
$A\mapsto Z(\cdot\,;A)$ is analytic on an open set of drives, and the static spectral map of the companion paper~\cite{Tufillaro2026a} is this family composed with the
output map: the $h$-th harmonic of the scattered wave is the $h$-th Fourier coefficient
of the port output on $Z(\theta;A)$. The equivariance \eqref{eq:equivariance} gives
\begin{equation}
  Z(\theta;e^{j\varphi}A)=Z(\theta+\varphi;A)\ \Rightarrow\
  Z_h(e^{j\varphi}A)=e^{jh\varphi}Z_h(A),
\end{equation}
the statement that harmonic $h$ carries weight $h$, now derived from the circuit
equations \eqref{eq:ode} rather than from the input--output map as in \citet{Tufillaro2026a}; and the same transport applies to the slow manifold,
\begin{equation}
  \Wc(e^{j\varphi}A)=\{(z,\theta): z\in W_{\theta+\varphi}(A)\},
  \label{eq:Wtransport}
\end{equation}
so the slow manifolds of drives that differ only in phase are one manifold shifted in
$\theta$. Each slow manifold contains its own LSOP, the memoryless model is the family of
these orbits taken over the drives, and the circle symmetry of the companion
paper~\cite{Tufillaro2026a} is a symmetry of the whole dynamic object. In the hidden-variable formulation of the
X-parameter book~\cite[Sec.~6.6.2]{Root2013}, the fixed-point family is the set of
quasi-static values $s_i(|A|)$ of the hidden variables, and the linearization about those
values is the linearized kernel of Section~\ref{sec:radius}.

\subsection{The reduced dynamics is the dynamic X-parameter kernel}
\label{sec:kernel}

Fix coordinates $\xi\in E\simeq\R^{q}$ on the fiber through the parameterization $K$ of
Theorem~\ref{thm:main}, extended over the bundle by the flow. The scattered wave along a
trajectory on $\Wc(A)$ is a function $B=\mathcal H(\xi,\theta;A)$ and $\xi$ obeys the
reduced dynamics, of stroboscopic form $\xi\mapsto R(\xi)$. \citet{Verspecht2009} model long-term memory by
\begin{equation}
  B(t)=F\bigl(A(t)\bigr)+\int_0^{\infty}G\bigl(A(t),A(t-u),u\bigr)\,du,
  \label{eq:verspecht}
\end{equation}
identifying $G$ from the transient of $B$ after a step of the envelope from
$A_{\mathrm{prev}}$ to $A$. In the present language that transient is, after the electrical transient, the trajectory on
$\Wc(A)$ that starts at the slow coordinates of the fixed point of $\Wc(A_{\mathrm{prev}})$: the step response is a
flow line of the reduced dynamics between two points of the fixed-point family,
$G(A,A_{\mathrm{prev}},u)$ is minus the time derivative of $\mathcal H$ along it, and the
memoryless $F$ is $\mathcal H$ at the fixed point. The integral \eqref{eq:verspecht}
reproduces these step responses exactly and other envelopes
approximately~\cite{Verspecht2009}; only the step responses are used here. Three
consequences follow. First, $G$ is
nonlinear in $A_{\mathrm{prev}}$ exactly to the extent that the reduced dynamics is
nonlinear, and only in the limit $A_{\mathrm{prev}}\to A$ is it a sum of $q$ exponentials
with the slow Floquet exponents as rates. Second, the step responses sample the reduced
dynamics only along the flow lines that start on the fixed-point family. If the drive
amplitude alone is stepped, the starting points form a curve in the $q$-dimensional fiber
and the responses sweep a two-dimensional set, which for $q>2$ does not fill the fiber.
Steps in load and in bias add starting directions, and the fiber can be filled only when
the number of independent step parameters is at least $q-1$. Third, in a compact model the slow states are
physical (temperatures, trap voltages), so the fiber coordinates can be taken to be those
states to leading order in the ratio of rates; the parameterization of $\Wc(A)$ is then
the periodic steady state computed with the slow states pinned, and the reduced dynamics
is their period-averaged equation. That is the two-timescale computation of
Section~\ref{sec:example}, and Theorem~\ref{thm:main} is the statement that the object
computed in that way is well defined, and unique among analytic manifolds.

\subsection{Modulated drive}
\label{sec:modulated}

Theorem~\ref{thm:main} is stated at fixed $A$. An envelope-domain model uses the reduced
dynamics with the envelope $A(t)$ as a slowly varying parameter. The justification is
adiabatic: the envelope varies slowly compared with the electrical rates, so the
trajectory follows the family of slow manifolds. Proposition~\ref{prop:modulated} states
this, with a sketch of the proof.

\begin{proposition}[Modulated drive]\label{prop:modulated}
Let $A(s)$, $s\in\R/S\mathbb Z$, be a smooth periodic envelope such that the hypotheses of
Theorem~\ref{thm:main} hold for every $A(s)$, with the spectral gap between the slow
modes and the fast modes bounded below uniformly in $s$, and consider the drive
$A(\varepsilon t)$. Then there is $\varepsilon_0>0$ such that for $0<\varepsilon<\varepsilon_0$
the system $\dot z=G(z,\omega t;A(\varepsilon t))$ has a $(q{+}2)$-dimensional locally
invariant manifold $\Wc_\varepsilon\subset\R^N\times\Sone\times(\R/S\mathbb Z)$, normally
attracting at the electrical rates, which is $O(\varepsilon)$-close in
$C^{1}$ to the frozen family $\{(z,\theta,s): z\in W_\theta(A(s))\}$. The manifold is not
unique. The flow on
$\Wc_\varepsilon$ is, in the fiber coordinates $\xi$ of the frozen family,
\begin{equation}
  \dot\xi=r\bigl(\xi,\theta;A(s)\bigr)+O(\varepsilon),\qquad \dot s=\varepsilon,
  \label{eq:reduced_mod}
\end{equation}
with $r$ the reduced vector field of the frozen system at drive $A(s)$.
\end{proposition}

\begin{proof}This is a sketch.
Extend the phase space by the slow phase $s$ and pass to the stroboscopic map of the
extended system,
$(z,s)\mapsto(\Phi_T(z;A(s+\varepsilon\,\cdot)),\,s+\varepsilon T)$. At $\varepsilon=0$ it is
$(z,s)\mapsto(P_{A(s)}(z),s)$, for which the union over $s$ of the fibers $W_0(A(s))$,
each taken as a closed neighborhood of the fixed point in $W_0(A(s))$, is a compact
invariant manifold with boundary, which the map carries into itself because the fiber
dynamics is attracted to the fixed point. The manifold is normally attracting
because, by the second point after Theorem~\ref{thm:main},
the normal multipliers are those of the fast modes, of modulus at most
$e^{-\lambda_{\mathrm{el}}T}$ with $\lambda_{\mathrm{el}}>0$ the slowest fast decay rate, while the tangential multipliers are those of the slow
modes and $1$ (the $s$ direction). Fenichel's persistence
theorem~\cite[Theorem~1]{Fenichel1971} is stated for flows and for manifolds that are
overflowing invariant, that is, manifolds whose boundary the flow leaves. It is brought to bear
in the usual way~\cite{Fenichel1979,Kaper2001}: the map is suspended to a flow and modified
outside a smaller neighborhood of the fixed points, in the slow directions only, so that
the boundary becomes overflowing while the normal rates still dominate. The map for
$\varepsilon>0$ is $O(\varepsilon)$-close in $C^{1}$ on the compact set, so the theorem
gives a nearby manifold, invariant for the modified system and locally invariant for the
original one, attracting, $C^{1}$-$O(\varepsilon)$-close, and smooth of a class limited by
the ratio of normal to tangential rates. It depends on the modification, which is why it
is not unique. Sweeping it by the flow over one drive period gives
$\Wc_\varepsilon$, and \eqref{eq:reduced_mod} is the restriction of the flow to it
written in the coordinates of the frozen family.
\end{proof}

The proposition says what an envelope-domain model is: the frozen reduced dynamics driven
by the envelope, to first order in the ratio of envelope rate to electrical rate
($10^{-2}$ to $10^{-3}$ for a $10$\,MHz envelope). The envelope may be as fast as the slow
modes, or faster; what is required is that it be slow compared with the modes outside $E$.
The proposition is the periodic-orbit counterpart of the adiabatic spectral submanifolds
of \citet[Theorem~6]{Haller2024}, who treat a slowly varying
system whose frozen states are fixed points, with a full proof and without requiring the
envelope to be periodic. Periodicity is assumed here for
compactness, and it is the form of the test signals used by modulated-drive network
analyzers, sold as vector component analyzers (VCAs)~\cite{Verspecht2022}.

\section{The Bundle in Time-Domain Coordinates}
\label{sec:bundle}

A time-domain behavioral model~\cite{Wood2004,Wood2005} does not use phasors. It
samples a port waveform, forms delay vectors, and fits a map on the reconstructed state,
appealing to Takens' theorem for the embedding.
Takens' theorem is stated for autonomous systems, and a driven device is not autonomous.
The applicable statement is Stark's delay-embedding theorem for forced
systems~\cite{Stark1999}, whose setting is the bundle of
Theorem~\ref{thm:main}(c). Spectral submanifolds have been identified from decaying
transients in delay coordinates by \citet{Cenedese2022}, and the
tangent spaces of delay-embedded invariant manifolds characterized by \citet{Axas2023}, for autonomous systems at a fixed point. The setting here
differs in that the orbit is forced and the manifold is a bundle over the drive phase.

Stark considers a skew product: a forcing system $g$ on a manifold $N$ and a forced
system $f(\cdot,y)$ on a manifold $M$ of dimension $m$, observed through a function
$\phi:M\to\R$ that does not depend on the forcing state. For a periodically forced
vector field sampled at interval $\tau$ (period normalized to one), $N=\Sone$ is the
circle of drive phase and $g$ is the rotation by $\tau$; Stark's Theorem~3.3 states that for
an open dense set of vector fields and observables the delay map
$(x,\theta)\mapsto(\phi(x),\phi(f^{\tau}(x,\theta)),\dots,\phi(f^{(d-1)\tau}(x,\theta)))$
is an embedding of the whole of $M\times\Sone$ when $d\ge2m+3$ and $k\tau\notin\mathbb Z$
for $1\le k\le d$, and an embedding of the $k_0$ fibers $M\times\{\theta_0,\dots,\theta_{k_0-1}\}$
when $\tau=1/k_0$ and $d\ge2m+1$; sampling once per period, $\tau=1$, reduces to
Takens' theorem for the stroboscopic map on a single fiber.

Stark's theorem is stated for a compact manifold $M$ and a vector field on
$M\times\Sone$. Here the fiber is a compact neighborhood of the fixed point in $W_0$, a
disc that the flow carries into itself, so the theorem applies only after the restricted
flow has been extended to a compact manifold without boundary that contains the disc. That
extension is assumed here and is not carried out. With it, $m=q$, $f(\cdot,\theta)$ is the
flow of \eqref{eq:ode} restricted
to $\Wc(A)$, and $\phi$ is a node voltage of the circuit, which is a
function of the state alone. After the electrical transient has died, the port waveform
is a time series of such an observable on $\Wc(A)$, and the count of delays is as follows.

\begin{remark}[Number of delays at constant drive]\label{cor:embed}
Under the hypotheses of Theorem~\ref{thm:main}, for generic vector field and observable,
and granted the extension above,
$2q+1$ delays of a port waveform sampled once per drive period suffice to embed the fiber
$W_0$, and $2q+3$ delays sampled at an interval incommensurate with the period suffice to
embed the whole bundle $\Wc(A)$, where $q$ is the number of slow modes.
\end{remark}

The remark gives a model-order rule of the same kind as the coefficient count and the
truncation order in $r$ of the companion paper~\cite{Tufillaro2026a}: at constant drive, a number
of delays sufficient for a time-domain memory model is
set by the number of slow exponents, and the intrinsic dimension recovered from a decaying
transient by false-nearest-neighbor methods, as in \citet{Wood2005}, is an estimate
of $q$. Under a modulated drive the state space of Proposition~\ref{prop:modulated}
includes the envelope, and a time-domain model is then built on delays of the input as
well as the output~\cite{Walker1999}; the count above does not cover that case. We want to
emphasize two features of the remark. First, the
distinction between embedding a fiber and embedding the bundle is the distinction between
an envelope-domain model, which is defined on the fiber and treats the drive phase as
fixed, and a full time-domain model, which must also reconstruct the drive phase, one
more dimension, and so requires two more delays.
Second, the equivariance \eqref{eq:Wtransport} says that the slow manifolds of drives that
differ only in phase are time shifts of one another, so a
model identified at one phase of the drive applies at every phase: the time-domain
form of the gauge freedom of the companion paper~\cite{Tufillaro2026a}. (At a fixed drive the
bundle over the circle of $\theta$ is trivial for a separate reason: the slow multipliers
are real and positive, so the return map on the fiber does not reverse orientation.)

\section{The Hypotheses Examined on a GaN HEMT with Memory}
\label{sec:example}

\subsection{The device and the two-timescale computation}

The example is the ASM-HEMT compact model (version 101.4.0 of the model described
in \citet{Khandelwal2018};
$0.25\,\mu$m gate length, $0.8$\,mm gate width, Verilog-A compiled with OpenVAF and run in ngspice~44 through its OSDI
interface) at $2.45$\,GHz, $V_{DD}=15$\,V, class AB ($V_{gs}=-2.7$\,V, $I_{dq}\approx100$\,mA),
driven at the gate through $50\,\Omega$ and loaded at the drain by an incident-wave
source (an A-pull in the sense of \citet{Tufillaro2026a}, the load presented as an incident wave
$A_2=0.3+0.1j\,\sqrt{\mathrm W}$ rather than as an impedance, found by a coarse scan for
output power) through an ideal bias tee, with two memory mechanisms added: a three-pole
Foster thermal network with the time constants measured by \citet{Kellogg2020} ($6\,\mu$s, $100\,\mu$s, $1.5$\,ms; $R_{\mathrm{th}}=20$\,K/W
in total, in the range reported for GaN on SiC~\cite{Gonzalez2023}), and the drain-lag
trap of \citet{Jardel2007} as fitted to ASM-HEMT by
\citet{Beleniotis2022} (capture time $1$\,ns, emission
time $\tau_e=6$\,ms; the trap voltage $V_t$ shifts threshold, mobility and access-region parameters
by their published laws). The diode gating of the trap is smoothed with a softplus of
scale $nV_T=26$\,mV so that $G$ is analytic; $nV_T\to0$ is the ideal peak detector of the
original models. The slow states are $(T_1,T_2,T_3,V_t)$, so $q=4$; the electrical states
are the node charges of the transistor and its terminations. The supply was chosen so
that $V_t$ stays within the range of the published trap fit; the bias therefore differs
from that of the companion paper~\cite{Tufillaro2026a} ($20$\,V, $V_{gs}=-1.5$\,V).

The LSOP was computed by two-timescale reduction: the slow states are pinned during
short transient runs (eight RF periods, the last four analyzed), and their period-averaged
equations --- the thermal network driven by the period-mean dissipated power, the trap by
the period-mean diode currents --- are iterated to a fixed point. At a gate drive of
$1.5$\,V source voltage this gives a junction-temperature rise of $30.4$\,K, $V_t=24.80$\,V (the
envelope peak of the drain voltage is $24.67$\,V), an output of $30.0$\,dBm
($31.3$\,dBm with the slow states at their cold values) and a drain current of $180$\,mA; the power
balance closes in the time domain to $0.05\%$. The slow maps --- dissipated power,
output wave, and drain waveform --- were then tabulated on a grid of drive $\times$
temperature $\times$ trap voltage ($440$ periodic steady states) and interpolated by a
tensor-product cubic spline, in the manner of the table-based Root model~\cite{Root1991},
giving an envelope-domain surrogate whose slow equations
integrate in milliseconds; this surrogate is the parameterization of $\Wc(A)$ described
in Section~\ref{sec:kernel}, computed rather than expanded.\footnote{The device
description, the simulator scripts, the tables and the MATLAB routines used in this
section are archived in \citet{Tufillaro2026toolkit}; version 1.2.0 adds the full-circuit runs of
Section~\ref{sec:fullcheck}, the multistart refit of Section~\ref{sec:radius} and the
SSMTool comparison of Section~\ref{sec:polyssm}.}

\begin{figure}[!t]
\centering
\includegraphics[width=\textwidth]{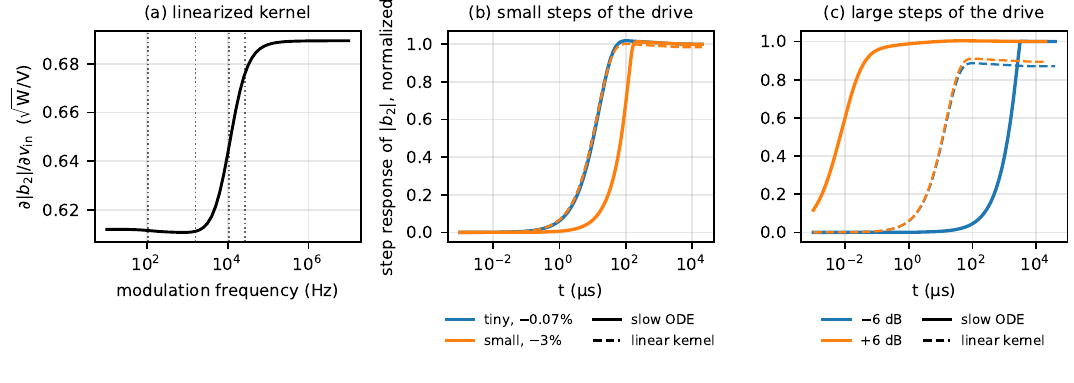}
\caption{ASM-HEMT with thermal and trap memory at $A_2=0.3+0.1j$, gate drive $1.5$\,V. (a) Linearized kernel
in the frequency domain: $\partial|b_2|/\partial v_{\rm in}$ against modulation frequency;
dotted lines at $|\lambda_i|/2\pi$ for the four slow exponents. (b), (c) Step responses
of $|b_2|$ on the surrogate slow dynamics (solid) against the linearized kernel (dashed),
each pair normalized to the final value of the slow-dynamics response: a $-0.07\%$ step follows the kernel; $-3\%$ and
$\pm6$\,dB steps do not, because the trap voltage slews instead of relaxing.}
\label{fig:kernel}
\end{figure}

\subsection{The spectrum}\label{sec:spectrum}

Table~\ref{tab:spectrum} lists the slow Floquet exponents, the eigenvalues of the
Jacobian of the averaged slow equations at the LSOP (finite differences of the dissipated
power and of the drain waveform, with the trap equation differentiated through the
waveform by the chain rule, since it is exponentially sensitive to the waveform peak).
These eigenvalues approximate the slow Floquet exponents of the full stroboscopic map to
leading order in the ratio of slow to electrical rates. The approximation was checked
directly for the fastest slow mode: with the thermal network placed in the circuit and the
trap voltage held at its LSOP value, a transient of the full circuit started from the
LSOP with the channel-temperature state displaced by $2$\,K decays with a time constant
of $6.015\,\mu$s ($24\,500$ RF periods simulated). The transient holds the choke current
fixed, so the DC drain voltage moves with the device current through the $50\,\Omega$
termination, and the averaged equations evaluated under the same condition give
$6.016\,\mu$s: agreement to $0.02\%$. (With the ideal tee of Table~\ref{tab:spectrum}
the mode is at $6.006\,\mu$s and the bare network pole at $6.000\,\mu$s; the loop gain is
$0.013$ with the choke current fixed against $0.005$ with the ideal tee, so the thermal
exponent depends on the DC termination.) The slowest electrical exponent is bounded by
the start-up transient of the same circuit: the deviation of the drain waveform from its periodic steady state falls by five
orders of magnitude within one RF period and reaches the numerical floor within three,
so the slowest electrical rate exceeds $5\times10^{9}\,\mathrm s^{-1}$; the node time
constants of the terminated transistor, tens of picoseconds, put the fast modes at
$10^{10}$--$10^{11}\,\mathrm s^{-1}$.

\begin{table}[!t]
\tbl{Floquet exponents of the LSOP with thermal and trap memory (real parts; all imaginary parts zero
modulo $j\omega$); electrical row: node-RC estimate, slowest rate bounded directly at
$5\times10^{9}\,\mathrm s^{-1}$.\label{tab:spectrum}}
{\begin{tabular}{lrrl}
\toprule
mode & $\Rel\lambda$ (s$^{-1}$) & $\tau$ & character\\
\colrule
1 & $-665$ & $1504\,\mu$s & thermal (mount) \\
2 & $-9\,974$ & $100.3\,\mu$s & thermal (die) \\
3 & $-68\,706$ & $14.6\,\mu$s & trap, linearized \\
4 & $-166\,510$ & $6.01\,\mu$s & thermal (channel) \\
electrical & $-10^{10}$ to $-10^{11}$ & $10$--$100$\,ps & node RC time constants \\
\botrule
\end{tabular}}
\end{table}

The thermal exponents are the open-loop poles of the network to a fraction of a percent;
the electrothermal loop gain $R_{\mathrm{th}}\,\partial\langle P\rangle/\partial T_j$ is
$+0.005$ at this LSOP.\footnote{A first estimate of $+0.009$ by a two-point difference
was corrupted by the tolerance of the bias-tee solve ($10^{-4}$\,A in the choke current is
$1.5$\,mW in dissipated power); the value quoted is consistent between the table
surrogate, the polynomial fit of Section~\ref{sec:polyssm} and a direct difference at a
tightened tolerance.} The trap exponent is not the emission rate
$1/\tau_e=1/(6\,\mathrm{ms})\approx170\,\mathrm s^{-1}$.
Linearized about an RF operating point, the exponential diode pins $V_t$ to the envelope
peak with rate $\langle(V_t-v_d)^{+}\rangle/(nV_T\,\tau_e)\approx9.8\,\mathrm V/(26\,\mathrm{mV}\times6\,\mathrm{ms})=6.3\times10^{4}\,\mathrm s^{-1}$,
within $10\%$ of the $6.9\times10^{4}\,\mathrm s^{-1}$ the Jacobian returns, and it scales with $nV_T$ ($28\,\mu$s at $n=2$,
$6.2\,\mu$s at $n=0.5$). The emission time is a large-signal property of the trap and does
not appear in the spectrum.

\subsection{The hypotheses}

Of the hypotheses of Theorem~\ref{thm:main}, the slow exponents are computed, from the
averaged slow equations, and two of them, the fastest thermal mode and the trap mode, are
checked on the full circuit (Sections~\ref{sec:spectrum} and~\ref{sec:fullcheck}). The fast exponents are bounded and are not computed.
Non-resonance with the fast modes is argued below and cannot be computed. Analyticity is a
property of the model equations.

\emph{Stability.} All four slow exponents are real and negative; the fast modes are
damped by the terminations, as the start-up transient of Section~\ref{sec:spectrum} shows
for the modes it excites. Assumption~\ref{ass:stable} holds to that extent.

\emph{Spectral quotient.} $\Sigma(E)=\Int[10^{10}\text{--}10^{11}/665]\approx1.5\times10^{7}$
to $1.5\times10^{8}$, and at least $7.5\times10^{6}$ from the measured bound on the electrical
rate.

\emph{Non-resonance of the full set of slow modes.} Condition \eqref{eq:nonres} asks whether
an integer combination of the four slow rates, of order between $2$ and $\Sigma(E)$,
equals an electrical rate exactly. Combinations of order below $3\times10^{4}$ cannot reach the
fast modes; at orders from $3\times10^{4}$ to $\Sigma(E)$ near-coincidences necessarily occur, but an
exact one is a measure-zero event, and a near one affects only Taylor coefficients of
that order, which are of no practical consequence (Section~\ref{sec:radius}). Assumption~\ref{ass:nonres} is taken to hold on that ground, and
Theorem~\ref{thm:main} then gives a four-dimensional fiber, a five-dimensional bundle $\Wc(A)$,
and reduced dynamics with four states. By Remark~\ref{cor:embed}, nine delays of a port
waveform sampled once per period suffice to embed the envelope fiber at constant drive, and eleven suffice
for the bundle when the sampling interval is incommensurate with the period; the
intrinsic dimension a reconstruction should recover is four.

\emph{Internal resonances.} $\Rel\lambda_2/\Rel\lambda_1=14.999$. The bare thermal
network has poles in the exact ratio $15$, from the round time constants of the thermal
fit, and the electrothermal loop detunes the ratio by $7\times10^{-5}$, which is less than
the accuracy to which the Jacobian is known. $\Rel\lambda_3/\Rel\lambda_2$ is $6.89$ from
the Jacobian of Table~\ref{tab:spectrum} and $7.07$ from that of
Section~\ref{sec:polyssm}, a small divisor at order 7 in either case. Internal resonances
are permitted. An exact one would add to the polynomial $R$ of
Theorem~\ref{thm:main}(d) one resonant monomial, of degree 15; a near one leaves $R$
linear at the price of a large coefficient in the parameterization at that order.
Section~\ref{sec:polyssm}, point (iv), gives the size of that coefficient for this device.

\emph{Subsets of the slow modes.} Corollary~\ref{cor:count} applied to the slowest mode alone treats mode
2 as an outside exponent, within the four-state reduced dynamics, where the quotient is
$\Int[\Rel\lambda_4/\Rel\lambda_1]=250$. The condition $m\,\Rel\lambda_1\ne\Rel\lambda_2$
at $m=15$ holds by the margin of $7\times10^{-5}$ noted above, which is too small to be
confirmed numerically. A one-state thermal model
of this device is therefore covered by the theorem only marginally. A thermal network
with measured time constants would not sit this close to an integer ratio, and the
near-resonance affects only a Taylor coefficient of order 15
(Section~\ref{sec:polyssm}, point (iv)).

\subsection{The linearized kernel and its range of validity}\label{sec:radius}

The reduced dynamics on the surrogate was linearized and integrated for envelope steps
(Fig.~\ref{fig:kernel}). The linearized kernel --- the $A_{\mathrm{prev}}\to A$ limit of
\eqref{eq:verspecht}, computed as $Ce^{Ju}B_A$ from the Jacobian $J$ of the reduced
dynamics --- has step weights, for $|b_2|$ per volt of gate drive, of
$+1.2\times10^{-3}$, $+1.2\times10^{-3}$, $-8.0\times10^{-2}$ and $+5\times10^{-4}\ \sqrt{\mathrm W}/\mathrm V$
on the four modes of Table~\ref{tab:spectrum}, on top of a memoryless part of
$0.690\,\sqrt{\mathrm W}/\mathrm V$; the static gain of the kernel, $0.612$, equals the slope
of the output along the fixed-point family computed independently, to four digits, which
verifies the linearization. The modulation-frequency response of $|b_2|$
(Fig.~\ref{fig:kernel}(a)) rises from $0.612$ below $1$\,kHz to $0.690$ above $1$\,MHz:
about $1$\,dB of memory, essentially all of it due to the trap, with a corner near $10$\,kHz.

A step of $-1$\,mV in the drive ($-0.07\%$) follows this kernel to within $3\%$
(Fig.~\ref{fig:kernel}(b)). A constant and four exponentials were fitted to that transient
from $200$ random starts, each rate drawn within a factor of ten of its Floquet value. Of
the $200$ fits, $140$ converge to the same minimum, with time constants of $1503$, $99.8$
and $14.5\,\mu$s for the three slower modes, within $0.5\%$ of Table~\ref{tab:spectrum};
the rest stop at larger residuals. The fourth time constant is not recovered: the fit
gives $7.9\,\mu$s against $6.0\,\mu$s, for a mode whose kernel weight is below $1\%$ of
the trap's. A fit with the four rates fixed at their Floquet values leaves a residual $70$
times that of the free fit, so at this step size the transient is not yet a sum of the
four linear modes to the accuracy that the weakest of them requires. A step of $-3\%$ is
already $68\%$ away from the linear prediction, and steps of $\pm6$\,dB are $85$--$96\%$
away (Fig.~\ref{fig:kernel}(c)): the trap voltage slews in place of relaxing, downward at
about $1.5$\,V/ms, close to the emission-limited rate
$\langle V_t-v_d\rangle/\tau_e\approx1.6$\,V/ms, and upward within about $100$\,ns, limited
by capture, while
the thermal modes stay linear throughout. The reason is the diode: the reduced vector
field depends on $V_t$ through $\exp((v_{d,\mathrm{peak}}-V_t)/nV_T)$, so its Taylor
expansion about the LSOP, and with it the polynomial $R$ of Theorem~\ref{thm:main}(d),
is informative only for $|\Delta V_t|\lesssim nV_T$, a few tens of millivolts, or a few
millivolts of drive.

This is the practical content of the theorem for this device. The slow manifold exists,
is unique among analytic manifolds, has four dimensions and is attracting; the small-modulation kernel is a sum
of four exponentials with the rates of Table~\ref{tab:spectrum}, the trap's $15\,\mu$s
among them a property of the operating point; and at any realistic modulation
depth the kernel \eqref{eq:verspecht} is nonlinear in $A_{\mathrm{prev}}$ because the
reduced dynamics is nonlinear. A model linear in the memory excursion with level-dependent
coefficients, the structure of the dynamic gain model~\cite{Verspecht2023}, describes the thermal memory of
this device at every step tested and fails for its trap memory beyond a few millivolts of
drive. The manifold's parameterization
must therefore be represented globally, by tables, splines or a fitted network, as the
surrogate does. The polynomial $R$ of Theorem~\ref{thm:main}(d), which the
parameterization method produces at the fixed point, does not serve for the trap.

\subsection{The reduced trap dynamics against the full circuit}
\label{sec:fullcheck}

Everything above rests on the two-timescale computation, in which the slow states are
pinned during short runs and their equations are averaged. For the trap, which carries
nearly all of the memory, the reduction was tested against transient simulations of the
full circuit, in which nothing is pinned and nothing is averaged. The trap voltage was
made a node of the netlist. A sixth terminal was added to the compact model, through
which $V_t$ sets the five trap-dependent parameters by the same published laws, and the
trap capacitor is charged and discharged by the same smoothed diode law. With that
terminal held by a voltage source the modified model reproduces the drain waveform of the
original to $10^{-7}$\,V in a swing of $20$\,V. The thermal network is in the netlist as
in Section~\ref{sec:spectrum}, and the choke current is held fixed; the reduced vector
field was evaluated under the same condition. The full circuit costs about $0.4$\,s of
computation per RF period, fifteen minutes per microsecond, so the tests use runs of $30$
to $120$\,ns and compare rates of change, which for a first-order equation at fixed
thermal state determine the trajectory.

\begin{figure}[!t]
\centering
\includegraphics[width=\textwidth]{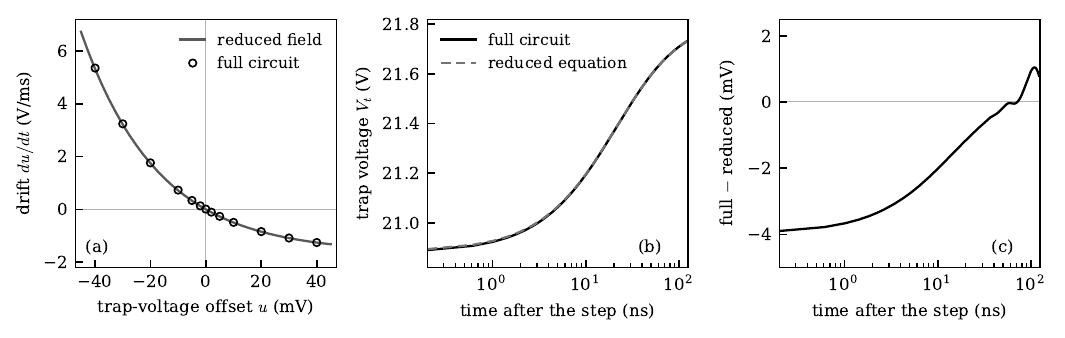}
\caption{The reduced trap dynamics against the full circuit, choke current fixed.
(a) Drift of the period-averaged trap voltage in the full circuit (circles), started at
offsets $u$ from its LSOP value at a gate drive of $1.5$\,V, and the reduced vector field
(line). (b) Trap voltage after a step of the drive from $0.75$\,V to $1.5$\,V: full
circuit (solid) and reduced equation (dashed). (c) Their difference.}
\label{fig:full}
\end{figure}

\begin{table}[!t]
\tbl{Drift $du/dt$ of the trap voltage, in V/s, at offsets $u$ from its LSOP value:
full circuit (period average, $80$ RF periods) and reduced vector field.\label{tab:fullcheck}}
{\begin{tabular}{lrrrrrrrr}
\toprule
$u$ (mV) & $-40$ & $-20$ & $-10$ & $-5$ & $+5$ & $+10$ & $+20$ & $+40$\\
\colrule
full circuit & $5349$ & $1760$ & $728$ & $336$ & $-265$ & $-493$ & $-842$ & $-1255$\\
reduced & $5311$ & $1749$ & $715$ & $325$ & $-272$ & $-499$ & $-846$ & $-1257$\\
\botrule
\end{tabular}}
\end{table}

\emph{Drift field.} At the LSOP drive the trap was started at $V_t^{\ast}+u$ for thirteen
offsets $u$ between $-40$ and $+40$\,mV, and the drift of its period average was read from
$80$ RF periods (Fig.~\ref{fig:full}(a), Table~\ref{tab:fullcheck}). Over that range the
drift runs from $+5.3$\,V/ms to $-1.3$\,V/ms, the asymmetry being the exponential diode.
The full circuit and the reduced field agree to within $1\%$ for $|u|\ge20$\,mV. Nearer
the fixed point they differ by a constant $9$\,V/s, which is a displacement of the fixed
point by $0.15$\,mV in $24.8$\,V. The slope at the fixed point, the trap exponent, is
$-5.95\times10^{4}\,\mathrm s^{-1}$ in the full circuit and $-5.96\times10^{4}\,\mathrm s^{-1}$
in the reduced field, a time constant of $16.8\,\mu$s in both. With the ideal tee the same
slope is $-6.6\times10^{4}\,\mathrm s^{-1}$, so the trap exponent depends on the DC
termination, as the thermal one does.

\emph{Steps of the drive.} With all states at the LSOP of the old drive, the drive was
stepped and the full circuit run on. After a step of $-3\%$ the trap voltage falls at
$1.264$\,V/ms in the full circuit and $1.266$\,V/ms in the reduced equation; after a step
of $-6$\,dB the two rates are $1.179$\,V/ms and agree to five digits. After a step of
$+6$\,dB from the LSOP at $0.75$\,V the trap voltage rises by $0.84$\,V in $120$\,ns, and
the reduced equation follows the full circuit to within $4$\,mV throughout
(Fig.~\ref{fig:full}(b),~(c)). The capture takes a few tens of RF periods, which is where
the separation of time scales behind the reduction is weakest; the largest difference is
at the start and is the capture of the first RF period.

These tests cover the trap direction, the nonlinearity included. A relaxation followed
over several time constants would take about twelve hours of computation for the trap and
weeks for the two slower thermal modes, and was not run.

\subsection{Explicit polynomial reduced models on the thermal submanifolds}
\label{sec:polyssm}

Theorem~\ref{thm:main}(d) guarantees a polynomial reduced dynamics, and the
parameterization method shows how to construct it. It was carried out for the two
submanifolds of the slow dynamics that a thermal-memory model would use: the
two-dimensional manifold tangent to the $1.5$\,ms and $100\,\mu$s modes and the
one-dimensional manifold tangent to the slowest mode alone. Because the spline surrogate is
only $C^{2}$, its Taylor coefficients beyond order three are artifacts; the slow vector
field was therefore re-fitted as a degree-four polynomial in $(T_j,V_t)$ to $49$ periodic
steady states on a Chebyshev box of $\pm12$\,K and $\pm0.12$\,V about the LSOP (dissipated
power to $10^{-4}$\,W rms; the drain waveform sample by sample), with the trap's softplus
gating kept exact and expanded as a power series about the fixed point, giving the Taylor
polynomial of the slow field to order seven. Its Jacobian reproduces
Table~\ref{tab:spectrum} ($1504$, $100.3$, $14.2$, $6.01\,\mu$s; $\lambda_2/\lambda_1=14.999$,
$\lambda_3/\lambda_2=7.07$); the $3\%$ difference on the trap mode arises because its rate
is exponentially sensitive to the waveform peak, which the two Jacobians differentiate
differently (chain rule through the spline waveform against the polynomial fit). The invariance equation $DW(\xi)R(\xi)=f(W(\xi))$ was then solved
order by order, with $R$ kept linear unless an inner resonance
appears (Table~\ref{tab:ssm}). As a check by an independent code, the same order-seven
field was given to SSMTool~2.6 \cite{SSMTool,Jain2022}. Its normal-form parameterization of both
manifolds agrees with the one tabulated here coefficient by coefficient, to relative
differences below $10^{-12}$ for the two-dimensional manifold and below $10^{-11}$ for the
one-dimensional one.

\begin{table}[!t]
\tbl{Parameterization of the thermal submanifolds, order by order: largest
coefficient of $W$ at that order (eigen-coordinates, units of kelvin) and the smallest
relative divisor $|\langle k,\lambda_E\rangle-\lambda_i|/|\lambda_i|$ with the mode $i$ at
which it occurs.\label{tab:ssm}}
{\begin{tabular}{crlcrl}
\toprule
\multicolumn{3}{c}{2-D manifold, modes 1,2} & \multicolumn{3}{c}{1-D manifold, mode 1}\\
\cline{1-3}\cline{4-6}
order & $\max|W_k|$ & divisor (mode) & order & $\max|W_k|$ & divisor (mode)\\
\colrule
2 & $6.6\times10^{-4}$ & $0.067$ (2) & 2 & $2.4\times10^{-5}$ & $0.87$ (2)\\
3 & $9.9\times10^{-6}$ & $0.13$ (2)  & 5 & $1.4\times10^{-10}$ & $0.67$ (2)\\
4 & $7.9\times10^{-7}$ & $0.20$ (2)  & 10 & $3.6\times10^{-17}$ & $0.33$ (2)\\
5 & $6.7\times10^{-9}$ & $0.27$ (2)  & 14 & $5.3\times10^{-23}$ & $0.067$ (2)\\
6 & $1.4\times10^{-9}$ & $0.15$ (3)  & 15 & $1.7\times10^{-21}$ & $7.2\times10^{-5}$ (2)\\
7 & $3.5\times10^{-9}$ & $0.010$ (3) & 16 & $8.0\times10^{-25}$ & $0.067$ (2)\\
\botrule
\end{tabular}}
\end{table}

\begin{figure}[!t]
\centering
\includegraphics[width=\textwidth]{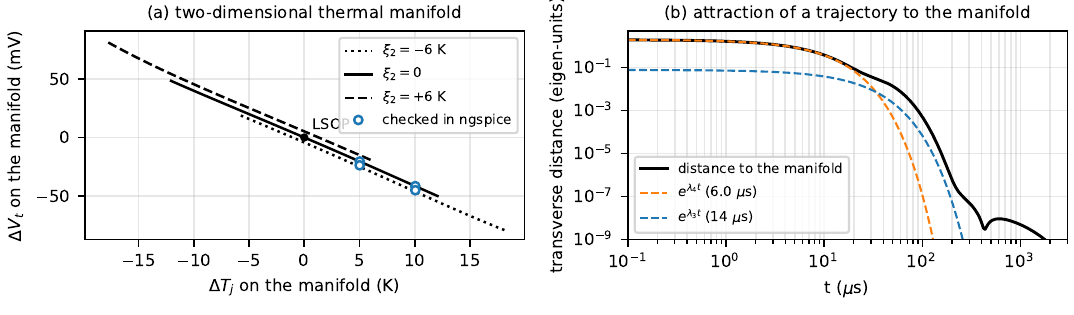}
\caption{The two-dimensional thermal manifold of Table~\ref{tab:ssm}. (a) Trap-voltage
excursion against junction-temperature excursion along the manifold: nearly a plane, with
the trap slaved to the temperature at $-4$\,mV/K. Circles: points at which ngspice steady
states confirm the tangency of the true vector field. (b) Transverse distance from the
manifold of a trajectory of the four-dimensional slow system started $2$ eigen-units off
it, with the two transverse decay rates for comparison; the floor near $10^{-8}$ is the
truncation of the order-seven parameterization.}
\label{fig:ssm}
\end{figure}

Four points can be made about Table~\ref{tab:ssm} and Fig.~\ref{fig:ssm}. (i) The two-dimensional manifold is nearly flat and its reduced
dynamics is linear to order seven: no inner resonance occurs, the coefficients of $W$
fall from $7\times10^{-4}$ at order two to $10^{-9}$ at order seven, the invariance
residual on the polynomial field is $3\times10^{-5}$ of $|f|$ at a $10$\,K excursion, and
the manifold's trap coordinate follows the temperature at $-4$\,mV/K --- the slaving of the
trap to the thermal state (Fig.~\ref{fig:ssm}(a)). The order-7 small divisor of $1.0\times10^{-2}$ against mode 3
is the near-resonance $7\lambda_2\approx\lambda_3$ noted above, and it has no practical effect:
the coefficient it amplifies is $3\times10^{-9}$. (ii) The invariance holds for the
compact model as simulated as well as for the polynomial fitted to it: at four points on the manifold with excursions of $5$--$10$\,K,
periodic steady states computed in ngspice give a vector field whose tangential components
agree with $DW\,R$ to $0.2\%$ and whose transverse components, divided by the transverse
rates, correspond to displacements of $10^{-3}$--$10^{-2}$ eigen-units, the numerical noise
floor of the simulation. (iii) The manifold attracts: a trajectory of the four-dimensional
slow system started $2$ eigen-units off it is $5\times10^{-4}$ away at $100\,\mu$s and
$5\times10^{-8}$ at $300\,\mu$s, the fast modes having decayed at their own rates
(Fig.~\ref{fig:ssm}(b)), after which it moves along the manifold. (iv) The one-dimensional manifold exhibits the order-15
near-resonance: the divisor $|15\lambda_1-\lambda_2|/|\lambda_2|$ is $7\times10^{-5}$, and the
coefficient at that order is $30$ times the one at order 14, against a trend that falls by a
factor of 20--30 per order, an amplification of about $10^{3}$. The hypothesis of the
theorem holds for this subset of the modes by that margin only. The amplified coefficient
is nonetheless $10^{-21}$ in
kelvin per kelvin$^{15}$, which contributes $10^{-5}$\,K at the $12$\,K edge of the
expansion box; the near-resonant term becomes comparable to the linear one only near a $30$\,K
excursion, beyond the range in which the local expansion has any validity.

Point (i) is a statement at this drive. The ratio $\Rel\lambda_3/\Rel\lambda_2$ follows the
trap exponent, which moves with the operating point (Section~\ref{sec:opdep}). From the
values of Table~\ref{tab:opdep}, which carry the $3\%$ uncertainty of the trap exponent
noted above, it is $3.98$, $5.04$, $6.01$, $6.87$ and $7.72$ at the five drives. Between
$0.75$\,V and $1.75$\,V it therefore passes through $5$, $6$ and $7$, and it lies at $4$
to within the uncertainty at the lowest drive. Near each crossing a coefficient of the
two-dimensional parameterization at that order is divided by a small number, as at order
seven above, and at the crossing itself Theorem~\ref{thm:main} gives no two-dimensional
thermal manifold. The non-resonance condition for the two thermal modes alone thus fails
at isolated drives within the operating range, and Proposition~\ref{prop:modulated} cannot
be applied to the thermal sub-manifold alone along an envelope that crosses them. For the
full set of four slow modes these are internal resonances, which are permitted.

The trap direction behaves differently. The same order-seven Taylor field, evaluated
along $V_t$ at fixed temperature, reproduces the exact trap equation to $10^{-3}$ at
$\pm50$\,mV, to $3\%$ at $\pm80$\,mV and fails at $\pm120$\,mV ($37\%$): the radius of the
expansion in the trap direction is two to three times $nV_T$, as anticipated. The thermal
submanifolds are well behaved because they lie along directions in which the device is
nearly linear and the trap is slaved; the nonlinearity that limits the polynomial
description is transverse to them. A polynomial reduced model is therefore adequate for
thermal memory at this operating point and inadequate for trap memory, which is where the
global representation of Section~\ref{sec:radius} is needed.

\subsection{Dependence on the operating point}
\label{sec:opdep}

Everything above is at one operating point. Table~\ref{tab:opdep} gives the operating
point, the spectrum and the linearized kernel at five further operating points: four
other drives at the same bias, obtained from the same tables, and the same drive at a
second gate bias ($V_{gs}=-2.9$\,V, $I_{dq}\approx40$\,mA, deeper class AB), obtained by
a fresh two-timescale computation. The three thermal exponents
are the network's own poles at every point, moving by less than $0.4\%$, because the
electrothermal loop gain is at most $0.009$ and changes sign at low drive. The trap
exponent is the one that moves: from $25\,\mu$s at $0.75$\,V to $13\,\mu$s at $1.75$\,V,
following $\langle V_t-v_d\rangle/(nV_T\tau_e)$ as the trap voltage rises with the
envelope peak, and to $16\,\mu$s at the deeper bias, where the peak is lower. The trap
accounts for nearly all of the memory at every operating point: its kernel weight is
$30$--$200$ times each thermal weight, and the memory fraction of the small-signal gain, one minus the ratio of static
to instantaneous slope, is $10$--$13\%$ until the device saturates at $1.75$\,V, where the
instantaneous slope itself falls by a factor of two. The required model order is therefore
four at every
operating point examined, the thermal exponents are the same at all of them, and
what an operating-point-dependent memory model must track is the trap exponent and the
trap weight.

\begin{table}[!t]
\tbl{Operating-point dependence: gate drive $v_{\rm in}$ (source voltage), junction-temperature
rise, trap voltage, output power, trap Floquet time constant $\tau_3$ (the thermal ones
stay at $1501$--$1506$, $100.1$--$100.4$ and $6.00\,\mu$s), electrothermal loop gain,
memoryless slope $D=\partial|b_2|/\partial v_{\rm in}$, trap kernel weight $w_3$ and
static slope, all in $\sqrt{\rm W}$/V. Last row ($^{\ast}$): $V_{gs}=-2.9$\,V; kernel not
computed, the tables being for $V_{gs}=-2.7$\,V.\label{tab:opdep}}
{\begin{tabular}{crrrrrrrr}
\toprule
$v_{\rm in}$ & $\Delta T_j$ & $V_t$ & $P_{\rm out}$ & $\tau_3$ & $R_{\rm th}\,\partial_{T}\langle P\rangle$ & $D$ & $w_3$ & static\\
(V) & (K) & (V) & (dBm) & ($\mu$s) & & & & \\
\colrule
0.75 & 31.5 & 20.89 & 23.7 & 25.2 & $-0.001$ & 0.762 & $-0.098$ & 0.662\\
1.00 & 32.1 & 22.29 & 26.5 & 19.9 & $+0.000$ & 0.734 & $-0.095$ & 0.639\\
1.25 & 31.7 & 23.59 & 28.5 & 16.7 & $+0.003$ & 0.711 & $-0.088$ & 0.624\\
1.50 & 30.4 & 24.80 & 30.0 & 14.6 & $+0.005$ & 0.690 & $-0.080$ & 0.612\\
1.75 & 28.2 & 25.96 & 31.3 & 13.0 & $+0.009$ & 0.330 & $-0.036$ & 0.297\\
\colrule
$1.50^{\ast}$ & 27.5 & 24.00 & 28.2 & 16.0 & $+0.004$ & --- & --- & ---\\
\botrule
\end{tabular}}
\end{table}

\section{Discussion}

A device with memory at
a large-signal operating point has, under the hypotheses of Theorem~\ref{thm:main}, an attracting, finite-dimensional invariant
manifold carrying its long-term memory, unique among analytic manifolds; the fiber dimension is the number of slow Floquet
exponents, which is the number of states an exact envelope-domain model needs and the
intrinsic dimension a time-domain model reconstructs at constant drive; the memoryless model is the
family of its fixed points (the LSOPs, one periodic orbit per drive); the step responses behind the dynamic X-parameter kernel are flow lines of its reduced dynamics between fixed
points; and the whole object carries the circle symmetry of the memoryless theory.

Memory, usually described dynamically, by a kernel or by a fitted map, is thus also an
object of geometry: a manifold fixed by the operating point, on which the motion is slow
compared with the electrical rates. The memoryless device is formally the case $q=0$, for which
the manifold is the periodic orbit alone; memory adds to that orbit a fiber of dimension
$q$ at each drive phase, transverse to it, and changes nothing else in the construction.

Five limitations should be mentioned. The theorem does not give the size
of the neighborhood on which the manifold is invariant, nor of the domain of its
parameterization; for the ASM-HEMT of Section~\ref{sec:example} the domain of the
parameterization is set by the exponential diode gating of the drain-lag trap and
extends only a few thermal voltages, $nV_T\approx26$\,mV,
in the trap-voltage direction (Section~\ref{sec:radius}). The theorem is stated at a
fixed drive; a modulated drive is covered only by Proposition~\ref{prop:modulated}, which
holds to first order in the ratio of the envelope rate to the electrical rates.
Quasi-periodic forcing, for which the operating point is an invariant torus rather than a
periodic orbit, is left to a later study. The theorem singles out the manifold by its smoothness class, but the other invariant
manifolds tangent to $E$ differ from it by terms of order $|\xi|^{\Sigma+1}$, far below
what measured data could resolve; in practice the manifold is singled out by attraction,
since every trajectory reaches it at the fast rates. The theorem also assumes the model is analytic
on the operating region, true of compact models by construction but not of every fitted
behavioral model. Last, the example is a simulation of one compact model in one simulator,
and its slow exponents are those of the averaged slow equations. Two of the four, and
the nonlinear drift of the trap, were checked on the full circuit; the two slower thermal
modes were not.

The natural next step is a measurement on a physical device: the small-modulation
kernel of a GaN HEMT at a well-characterized operating point. Vector component
analyzers~\cite{Verspecht2022} measure calibrated incident and scattered phasors under a
periodic multitone envelope with coherent averaging, which is what a small-modulation
measurement needs, provided the envelope period is long compared with the slowest time
constant to be resolved; the periods of $0.3$--$30\,\mu$s used in \citet{Verspecht2022}
to maximize dynamic range are shorter than the two slowest time constants of
Table~\ref{tab:spectrum}. The measurement would test
two predictions: that the rates measured are the slow Floquet exponents, and that the
trap rate moves with drive and bias as Table~\ref{tab:opdep} says rather than remaining
fixed at the trap's emission time constant.

Another direction for further study is computational: whether the slow-manifold picture
can be used to construct more efficient simulation algorithms or test protocols. The
construction of Section~\ref{sec:example} is a model-order reduction from a compact model
to an envelope model with a guaranteed number of states, and each of its steps, the LSOP
with the slow states pinned, the slow Jacobian and its Floquet exponents, and the
tabulation of the slow maps, is a routine a circuit simulator already performs. An
envelope-transient simulation carries the full circuit state and solves a
harmonic-balance problem at every envelope sample; the reduced model integrates $q$
ordinary differential equations with tabulated right-hand sides. On the measurement
side, the same picture fixes what a test of a device with memory should record: the slow
Floquet exponents at each operating point, from a small-modulation kernel, and, for a
time-domain model, a number of delays set by Remark~\ref{cor:embed}. Two questions
decide whether the reduction can be automated as a simulator feature. The first is the
growth of the tabulation grid with the number of envelope inputs (both incident waves,
bias and load). The second is the treatment of electrical memory: the bias-network and
video-bandwidth modes have rates between the thermal and the electrical ones and must be
assigned to $E$ once the envelope bandwidth reaches them. A drive whose envelope reaches a
mode that has been left outside $E$ lies outside Proposition~\ref{prop:modulated}. Both questions concern the size
of the object to be built and leave its existence untouched: Theorem~\ref{thm:main} settles the
existence and dimension of the slow manifold and its uniqueness among analytic manifolds, and what remains is the cost of
representing it over the operating range.

\section*{Acknowledgments}
This paper is in remembrance of David Broomhead (1950--2014), a wizard at applying pure
mathematical results to engineering applications.

The author used Claude (Anthropic) to edit for clarity, check algebra, and assist with
the figures and simulation scripts; the author reviewed and validated all AI-assisted
output and takes full responsibility for the content.

\section*{Data Availability}
The device description, the simulator scripts, the tables, the full-circuit runs and the
MATLAB routines that support the findings of this study are deposited at Zenodo
\cite{Tufillaro2026toolkit}, doi:10.5281/zenodo.24602291.

\bibliographystyle{ws-ijbc}
\bibliography{refs}

\end{document}